\documentclass[11pt]{article}
\usepackage{url}
\usepackage[pdfstartview=FitH,pdfpagemode=UseNone,backref,colorlinks=true,citecolor=blue,linkcolor=blue]{hyperref}
\usepackage{amsfonts}
\usepackage{latexsym,amssymb,amsmath,amsthm,mathrsfs}

\usepackage[usenames,dvipsnames,table]{xcolor}
\usepackage[margin=1.0in]{geometry}

\newtheorem{theorem}{Theorem}[section]

\newtheorem{lemma}[theorem]{Lemma}

\newtheorem{claim}[theorem]{Claim}

\newcommand{\dsum}{\displaystyle\sum}

\newcommand{\expect}{\mathbb{E}}

\newcommand{\St}[1]{S^t \cap f^{-1}({#1})}
\newcommand{\Sf}[1]{S \cap f^{-1}({#1})}

\newcommand{\Tl}[1]{T_{\leq {#1}}}
\newcommand{\Rl}[1]{R_{\leq {#1}}}
\newcommand{\Tll}[2]{T_{\leq {#1}}\cap f^{-1}({#2})}
\newcommand{\Rll}[2]{R_{\leq {#1}}\cap f^{-1}({#2})}

	\title{A note on the size of query trees}
		\author{ Shai Vardi%
		\thanks{California Institute of Technology, Pasadena, CA, USA. E-mail: {\tt  svardi@caltech.edu}. }}
\begin{document}


\maketitle
	
\begin{abstract}
	We consider query trees of graphs with degree bounded by a constant, $d$. 
	We give simple proofs that the size of a query tree is constant in expectation and $2^{O(d)}\log{n}$ w.h.p.
\end{abstract}
\section{Introduction}
Let $G=(V,E)$ be an undirected graph whose degree is bounded by a constant $d$. We assume that $|V|=n$ is large: $d\ll n$. Let $r:V\rightarrow[0,1]$ be a ranking function that assigns each vertex a real number between $0$ and $1$, uniformly at random. We call $r(v)$ $v$'s \emph{rank}. 
Vertex ranks induce an orientation of the originally undirected edges - if $r(v)\leq r(u)$, the edge is oriented from $v$ to $u$; in case of equality, the edge is bi-directional.

A query tree $T_v$  is the set of vertices that are reachable from $v$ after the edges have been oriented according to $r$ (strictly speaking, it is not necessarily a tree, but we use the term ``query tree'' for consistency with e.g.,~\cite{MRVX12,NO08}).

The aim of this note is to give a simple proof that the query tree has size $2^{O(d)}\log{n}$ w.h.p.

\begin{theorem}\label{thm:main}
		Let $G=(V,E)$ be a graph whose degree is bounded by $d$ and let $r:V\rightarrow[0,1]$ be a function that assigns to each vertex $v \in V$ a number between $0$ and $1$ independently and uniformly at random. Let $T_{max}$ be the size of the largest query tree of  $G$: $T_{max} = \max\{|T_v|:v \in V\}$.  Then, for $L=4(d+1)$, 
	$$\Pr[|T_{max}|> 2^{L}\cdot15L \log{n}]  \leq \frac{1}{n^2}.$$ 

	\end{theorem}
The proof of Theorem~\ref{thm:main} is based on a proof in~\cite{RV16}, and employs a \emph{quantization} of the rank function. 
Let $f$ denote a quantization of $r$ (in other words $r(u)\geq r(v) \Rightarrow f(u)\geq f(v)$); let $T_v^f$ denote the query tree with respect to $f$. Then $T_v \subseteq T_v^f$. Therefore it suffices to bound $|T_v^f|$.

In Section~\ref{sec:disc}, we give a brief discussion on  query trees. The reader is referred to~\cite{VardiPhD} for an introduction to local computation algorithms and role query trees play therein, and to \cite{NO08} for an introduction to query trees and their use in the analysis of sublinear approximation algorithms.

\section{Preliminaries} We denote the set $\{0, 1,\ldots, m\}$ by $[m]$. Logarithms are  base $e$.
Let $G=(V,E)$ be a  graph. For any vertex set $S \subseteq V$, denote by $N(S)$ the set of vertices that are not in $S$ but are neighbors of some vertex in $S$: $N(S)=\{N(v):v \in S\} \setminus S$. The \emph{length} of a path is the number of edges it contains.
 
For a set $S \subseteq V$ and a function $f:V \rightarrow \mathbb{N}$,  we use $\Sf{i}$ to denote the set $\{v \in S : f(v)=i\}$.

Let $G=(V,E)$ be a graph, and let $f:V\rightarrow \mathbb{N}$ be  some function on the vertices.
An \emph{adaptive vertex exposure procedure} $A$ is one that does not know $f$ a priori. $A$ is given a vertex $v \in V$ and $f(v)$; $A$  iteratively adds vertices from $V\setminus S$ to $S$: for every vertex $u$ that $A$ adds to $S$, $f(u)$ is revealed immediately after $u$ is added.   Let $S^t$ denote $S$ after the addition of the $t^{th}$ vertex. The following is a simple concentration bound whose proof is given for completeness.

\begin{lemma}\label{lem:poop}
	Let $G=(V,E)$ be a graph, let $L>0$ be some constant, let $c=15L$, and let $f:V\rightarrow[L]$ be a function chosen uniformly at random from all such possible functions. Let $A$ be an adaptive vertex exposure procedure that is given a vertex $v \in V$. 
	Then, for any $\ell \in [L]$,  the probability that there is some $t$, $c\log{n} \leq t \leq n$ for which 
	$|\St{\ell}|>\frac{2|S^t|}{L}$
	is at most $\frac{1}{n^4}$.
\end{lemma}

\begin{proof}
	Let $v_j$ be the $j^{th}$ vertex added to $S$ by  $A$, and let $X_{j}$ be the indicator variable whose value is $1$ iff $f(v_j)=\ell$. 
	For any $t\leq n$, $\expect\left[ \dsum_{j=1}^t  X_{j}\right]=\frac{t}{L} $. As
	$X_i$ and $X_j$ are independent for all $i \neq j$, by the Chernoff bound, for $c\log{n} \leq t \leq n$,
	$$\Pr\left[ \dsum_{j=1}^t  X_{j}> \frac{2t}{L}\right]  \leq e^{\frac{-t}{3L}} \leq e^{-5\log n}.$$
	A union bound over all possible values of $t: c\log{n} \leq t \leq n$ completes the proof.
\end{proof}

 \section{Expectation}

We first show that the expected size of a query tree is a constant depending only on $d$.
\begin{theorem} [\cite{Ona10}]
	Let $G=(V,E)$ be a graph whose degree is bounded by $d$ and let $r:V\rightarrow[0,1]$ be a function that assigns to each vertex $v \in V$ a number between $0$ and $1$ independently and uniformly at random. Let $T_v$ be the size of the query tree of some vertex $v \in V$. Then $\expect[|T_v|]  \leq e^d$, where the expectation is over the random choices of $r$.
	\end{theorem}

\begin{proof}
	Let $k>0$ be an integer. For any path of length $k$ originating from $v$, the probability that the path is monotone decreasing is $\frac{1}{(k+1)!}$. There at at most $d^{k}$ such paths. Hence, by the union bound, the expected number of monotone paths of length $k$ originating from $v$ is at most $\frac{d^{k}}{(k+1)!}$, and  the expected number of vertices in these paths is at most $\frac{(k+1)d^{k}}{(k+1)!} = \frac{d^{k}}{k!}$. 
	Therefore, the expected total number of vertices in monotone non-increasing paths is at most
	
	$$\dsum_{k=0}^{\infty}\frac{d^k}{k!} = e^d,$$
	which is an upper bound on the expected size of the query tree.
		\end{proof}

\section{Concentration}
For the concentration bound,
let $r:V \rightarrow [0,1]$ be a function chosen uniformly at random from all such possible functions. Partition $[0,1]$ into $L=4(d+1)$  segments of equal measure, $I_1, \ldots, I_L$. For every $v \in V$, set $f(v) = \ell$ if $r(v) \in I_\ell$  ($f$ is a quantization of $r$).  

Consider the following method of generating two sets of vertices: $T$  and $R$, where $T \subseteq R$.  For some vertex $v$, set $T=R=\{v\}$.  Continue inductively: choose some vertex  $w \in T$, add all $N(w)$ to $R$ and compute $f(u)$ for all $u \in N(w)$. Add the vertices $u$ such that $u \in N(w)$ and $f(u) \geq f(w)$ to $T$. The process ends when no more vertices can be added to $T$. $T$ is the query tree with respect to $f$, hence $|T|$ is an upper bound on the size of the actual query tree (i.e., the query tree with respect to $r$).
However, it is difficult to reason about the size of $T$ directly, as the ranks of its vertices are not independent. The ranks of the vertices in $R$, though, \emph{are} independent, as $R$ is generated by an adaptive vertex exposure procedure. 
  $R$ is a superset of $T$ that includes $T$ and its boundary, hence $|R|$ is also an upper bound on the size of the query tree.

We now define $L+1$ ``layers'' - $\Tl{0}, \ldots, T_{\leq L}$: $T_{\leq \ell}= T \cap \bigcup_{i=0}^{\ell} f^{-1}(i)$. That is, $\Tl{\ell}$ is the set of vertices in $T$ whose rank is at most $\ell$. (The range of $f$ is $[L]$, hence $\Tl{0}$ will be empty, but we include it  to simplify the proof.) 
\begin{claim} \label{lemmaclaim}
	Set $L=4(d+1)$, $c=15L$. 
	Assume without loss of generality that $f(v)=0$. Then for all $0 \leq i \leq L-1$,
	$$\Pr[|\Tl{i}|\leq 2^ic\log{n} \wedge |\Tl{i+1}| \geq 2^{i+1}c \log{n}]\leq \frac{1}{n^4}.$$
\end{claim}

\begin{proof}

For all $0\leq i \leq L$, let $\Rl{i} = \Tl{i} \cup N(\Tl{i})$. 
 Note that \begin{equation}\Rll{i}{i}=\Tll{i}{i},\label{eq:ll}
 \end{equation}because if there had been some $u \in N(\Tl{i}), f(u)=i$, $u$ would have been added to $\Tl{i}$.

Note that
	$|\Tl{i}|\leq 2^ic\log{n} \wedge |\Tl{i+1}| \geq 2^{i+1}c \log{n}$ implies that \begin{equation} |\Tll{i+1}{i+1}|>\frac{|\Tl{i+1}|}{2}.\label{e1}\end{equation}
	
	In other words,
	the majority of vertices $v \in \Tl{i+1}$ must have $f(v)=i+1$.

 Given $|\Tl{i+1}| > 2^{i+1}c \log{n}$, it holds that $|\Rl{i+1}| > 2^{i+1}c \log{n}$ because $\Tl{i+1} \subseteq \Rl{i+1}$. Furthermore, $\Rl{i+1}$ was constructed by an adaptive vertex exposure procedure and so the conditions of Lemma~\ref{lem:poop} hold for $\Rl{i+1}$.  
 From Equations~\eqref{eq:ll} and~\eqref{e1} 
 we get 
	
	\begin{align}
	\Pr[\left|\Tl{i}\right|\leq 2^ic\log{n} \wedge \left|\Tl{i+1}\right| \geq 2^{i+1}c \log{n}] &\leq \Pr\left[\left|\Rll{i+1}{i+1}\right| > \frac{\left|\Tl{i+1}\right|}{2}\right] \notag\\
	&\leq \Pr\left[\left| \Rll{i+1}{i+1}\right|  > \frac{2\left|\Rl{i+1}\right| }{L}\right] \notag\\
	& \leq \frac{1}{n^4}, \label{eq:fact2}\notag
	\end{align}
	where the second inequality is because  $|\Rl{i+1}|\leq (d+1)|\Tl{i+1}|$, as $G$'s degree is at most $d$; the last inequality is due to Lemma~\ref{lem:poop}.
\end{proof}

\begin{lemma}\label{lemma:main}
	Set $L=4(d+1)$. Let $G=(V,E)$ be a graph with degree bounded by  $d$, where $|V|=n$.
  For any vertex $v \in G$,   $\Pr\left[ T_v >2^{L}\cdot 15L\log{n}\right]<\frac{1}{n^3}$.
	
\end{lemma}

\begin{proof}To prove Lemma \ref{lemma:main}, we need to show that, for $c=15L$, 
	\begin{equation*}
	\Pr[|\Tl{L}|> 2^L c \log{n}]<\frac{1}{n^3}.
	\end{equation*}
	We show that for $ 0\leq i \leq L, \Pr[|\Tl{i}| > 2^ic \log{n}]< \frac{i}{n^4}$, by induction.
	For the base of the induction, $|S_0| = 1$, and the claim holds.
	For the inductive step, assume that $\Pr[|\Tl{i}|> 2^i c\log{n}]< \frac{i}{n^4}$.
	Then
	\begin{align*}
	\Pr[|\Tl{i+1}|> 2^{i+1}c\log{n}] &= \Pr[|\Tl{i+1}|> 2^{i+1}c\log{n} : |\Tl{i}|>2^{i}c\log{n}]\Pr[|\Tl{i}|>2^{i}c\log{n}]\\
	& + \Pr[|\Tl{i+1}|> 2^{i+1}c\log{n} : |\Tl{i}|\leq2^{i}c\log{n}]\Pr[|\Tl{i}|\leq 2^{i}c\log{n}].
	\end{align*}
	From the inductive step and
	Claim~\ref{lemmaclaim}, using the union bound, the lemma follows. 
\end{proof}

Applying a union bound over all the vertices gives the size of \emph{each} query tree is  $O(\log{n})$ with probability at least $1-1/n^2$, completing the proof of Theorem~\ref{thm:main}.

\section{Discussion}
\label{sec:disc}

Query trees were introduced by Nguyen and Onak \cite{NO08}, where they bounded their expected size. Mansour et al. \cite{MRVX12}, studying query trees in the context of \emph{local computation algorithms} ~\cite{RTVX11} (see~\cite{centlocal2017} for a recent survey), showed that their size is at most $O(\log{n})$ w.h.p. The proof presented above is adapted from~\cite{RV16} - the proof is simpler and more elegant than that of~\cite{MRVX12}. Furthermore, in order to generate the random order required in the proof, it suffices to have a random function $f:V\rightarrow [L]$, where $L$ is a constant. This, combined with the fact the relevant set is of size at most $O(\log{n})$ w.h.p., allows us to use a random seed of length only $O(\log{n})$ to generate such an $f$. See~\cite{RV16, VardiPhD} for details.

\paragraph{Acknowledgments} 
We thank Guy Even for suggesting that a short note such as this might be informative and for his useful comments.

\bibliographystyle{plain}\bibliography{Vardi_PhD_Bibliography}

\end{document}